\documentclass[12pt,a4paper,leqno]{amsart}

\title[Conormal distributions in the Shubin calculus]{Conormal distributions in the Shubin calculus of pseudodifferential operators}

\author[M. Cappiello]{Marco Cappiello}

\address{Department of Mathematics, University of Torino, Via Carlo Alberto 10, 10123 Torino, Italy.}

\email{marco.cappiello[AT]unito.it}

\author[R. Schulz]{Ren\'e Schulz}

\address{Leibniz Universit\"at Hannover, Institut f\"ur Analysis, Welfenplatz 1, D--30167 Hannover, Germany}

\email{rschulz[AT]math.uni-hannover.de}

\author[P. Wahlberg]{Patrik Wahlberg}

\address{Department of Mathematics, Linn{\ae}us University, SE--351 95 V\"axj\"o, Sweden}

\email{patrik.wahlberg[AT]lnu.se}

\usepackage{tikz,xifthen}
\usepackage{latexsym}
\usepackage{amsmath}
\usepackage{amssymb}
\usepackage{amsthm}
\usepackage{bbm}
\usepackage{amsfonts}
\usepackage{mathrsfs}
\usepackage{calc}
\usepackage{cite}
\usepackage{color}
\usepackage{graphicx}
\usepackage{enumerate}
\usepackage{setspace}

\numberwithin{equation}{section}          %Detta gr att man f�r
\newtheorem{thm}{Theorem}
\numberwithin{thm}{section}

\newcommand{\rubrik}{}
\newtheorem{prop}[thm]{Proposition}
\newtheorem{cor}[thm]{Corollary}
\newtheorem{lem}[thm]{Lemma}

\theoremstyle{definition}

\newtheorem{defn}[thm]{Definition}
\newtheorem{example}[thm]{Example}

\theoremstyle{remark}

\newtheorem{rem}[thm]{Remark}              %T o m hit r bara allmn
\newcommand{\Ker}{\operatorname{Ker}}

\newcommand{\ro}{\mathbb R}
\newcommand{\no}{\mathbb N}
\newcommand{\rr}[1]{\mathbb R^{#1}}
\newcommand{\nn}[1]{\mathbb N^{#1}}

\newcommand{\co}{\mathbb C}

\newcommand{\cl}{\mathrm{cl}}
\newcommand{\dd}{\mathrm {d}}

\newcommand{\eabs}[1]{\langle #1\rangle}

\newcommand{\Sp}{\operatorname{Sp}}
\newcommand{\GL}{\operatorname{GL}}
\newcommand{\M}{\operatorname{M}}
\newcommand{\On}{\operatorname{O}}

\newcommand{\dbar}{{{{\ \mathchar'26\mkern-12mu \mathrm d}}}}

\newcommand{\WF}{\mathrm{WF}}
\newcommand{\dist}{\operatorname{dist}}

\newcommand{\cS}{\mathscr{S}}
\newcommand{\cT}{\mathcal{T}}
\newcommand{\cV}{\mathcal{V}}
\newcommand{\cTp}{\mathcal{T}_{\psi_0}}

\newcommand{\cF}{\mathscr{F}}

\newcommand{\wt}{\widetilde}
\newcommand{\wh}{\widehat}

\def\la{\langle}
\def\ra{\rangle}
\newcommand{\leqs}{\leqslant}
\newcommand{\geqs}{\geqslant}

\begin{document}

\begin{abstract}
We characterize the Schwartz kernels of pseudodifferential operators of Shubin type by means of an FBI transform. 
Based on this we introduce as a generalization a new class of tempered distributions called Shubin conormal distributions. 
We study their transformation behavior, normal forms and microlocal properties.
\end{abstract}


\maketitle

%%%%%%%%%%%%%%%%%%%%%%%
\section{Introduction}
%%%%%%%%%%%%%%%%%%%%%%%

The theory of pseudodifferential operators has proven to be a powerful tool in many disciplines of mathematics. 
The space of conormal distributions was designed to contain the Schwartz kernels of pseudodifferential operators with H\"ormander symbols, see \cite[Chapter~18.2]{Hormander0}.
Conormal distributions are the starting point for the theory of Lagrangian distributions and Fourier integral operators \cite[Chapter~25]{Hormander0}, but it has also been studied in itself to a great extent, and it is essential in several theories, see e.g. \cite{Bony,MelroseAPS}.  A distribution $u$ defined on a smooth manifold is conormal with respect to a closed smooth submanifold if $Lu$ belongs to a certain Besov space locally for certain differential operators $L$ that depend on the submanifold. 

For the well-studied pseudodifferential operators on $\rr d$ with Shubin symbols \cite{Shubin1}, we are not aware of a concept corresponding to conormal distributions. 
In this paper we fill this gap by introducing a theory of conormal distributions with repect to linear subspaces of $\rr d$, adapted to Shubin operators. 
Recall that a Shubin symbol $a \in \Gamma_\rho^m$ of order $m \in \ro$ satisfies the estimates
\begin{equation*}
|\partial_x^\alpha\partial_\xi^\beta a(x,\xi)| \lesssim (1+|x|+|\xi|)^{m-\rho|\alpha+\beta|}, \quad (x,\xi) \in \rr d \times \rr d, \ \alpha, \beta \in \nn d,
\end{equation*}
where $0 \leqs \rho \leqs 1$. 

The key feature of the Shubin symbols that is difficult to describe by the standard techniques is the inherent isotropy, in particular that taking derivatives with respect to $x$ increases the decay in $\xi$. The tool that we employ to circumvent this issue is a version of the short-time Fourier transform, which is more suitable to isotropic symbols than the standard Fourier transform on which the classical theory is based.

Our work may be seen as phase space analysis of Shubin conormality. 
We extend Tataru's characterization \cite{Tataru} of the Schwartz kernels of pseudodifferential operators with $m=\rho=0$ to $0 \leqs \rho \leqs 1$ and order $m \in \ro$.  The behavior of the symbols with respect to derivatives and the order is reflected in phase space. 
 
Based on the characterization of the Schwartz kernels of Shubin operators, we define conormal tempered distributions on $\rr d$ with respect to a linear subspace and an order $m \in \ro$. To distinguish them from H\"ormander's notion of conormal distribution, we use the prefix $\Gamma$-conormal. The Schwartz kernels of Shubin operators are thus identical to the $\Gamma$-conormal distributions on $\rr {2d}$ with respect to the diagonal in $\rr {2d}$. 

We prove functional properties of $\Gamma$-conormal distributions and check that they transform well under the Fourier transform and linear coordinate transformations. We equip them with a topology such that these operators become continuous. 
The present paper can be seen as a first step in the direction of a phase space analysis for Lagrangian distributions in the Shubin calculus which, as far as we know, does not yet exist. This will be the subject of a forthcoming paper.

The paper is organized as follows: 
In Section \ref{sec:prelim} we introduce the FBI-type integral transform on which our analysis is based and state its basic properties. 
Section \ref{sec:shubchar} contains a phase space characterization of Shubin symbols in terms of the integral transform. 
In Section \ref{sec:pseudochar} we transfer the characterization to the Schwartz kernels of the associated class of global pseudodifferential operators. 
Along the way we give a simple proof of the continuity of these operators on the associated scale of Shubin--Sobolev modulation spaces. 
Finally in Section \ref{sec:gconorm} we define $\Gamma$-conormal distributions and discuss their functional and microlocal properties. 

%%%%%%%%%%%%%%%%%%%%%%%%%%%%%
\section{An integral transform of FBI type}\label{sec:prelim}
%%%%%%%%%%%%%%%%%%%%%%%%%%%%%

In this section we introduce the tool for the definition of Shubin conormal distributions, namely a variant of the FBI transform, and discuss its main properties. First we fix some notation. 

\subsection*{Basic notation}

We use $\cS(\rr d)$ and $\cS'(\rr d)$ for the Schwartz space of rapidly decaying smooth functions and its dual the tempered distributions. We write $\langle u,v\rangle$ for the bilinear pairing between a test function $v$ and a distribution $u$ and $(u,v)=\langle u,\overline{v}\rangle$ for the sesquilinear pairing as well as the $L^2$ scalar product if $u, v \in L^2(\rr d)$.

We use $T_{y}u(x) = u(x-y)$
and $M_\xi u(x) = e^{i \la x,\xi \ra}u(x)$, where $\la \cdot,\cdot \ra$ denotes the inner product on $\rr d$, for the operation of translation by $y\in \rr d$
and modulation by $\xi\in \rr d$, respectively, applied to functions or distributions. For $x \in \rr d $ we write $\eabs{x}=\sqrt{1+|x|^2}$. Peetre's inequality is
\begin{equation}
\label{eq:Peetre}
\eabs{x+y}^s \leqs C_s \eabs{x}^s\eabs{y}^{|s|}\qquad x,y \in \rr d, \quad s \in \ro, \quad C_s>0.
\end{equation}
We write $\dbar x$ for the dual Lebesgue measure $(2\pi)^{-d}\dd x$.
The notation $f (x) \lesssim g(x)$ means that $f(x) \leqs C g(x)$ for some $C>0$, for all $x$ in the domain of $f$ and $g$. 
If $f (x) \lesssim g (x) \lesssim f (x)$ then we write $f (x) \asymp g (x)$.

The Fourier transform is normalized for $f \in \cS(\rr d)$ as 
\begin{equation*}
\cF f (\xi) = \widehat f (\xi) = (2\pi)^{-d/2} \int_{\rr d} f(x) e^{-i \la x,\xi \ra} \, \dd x 
\end{equation*}
which makes it unitary on $L^2(\rr d)$. 
The partial Fourier transform with respect to a vector variable indexed by $j$ is denoted $\cF_j$. 
For $1 \leqs j \leqs d$ we use $D_j = -i \partial_j$ and extend to multi-indices. 

The orthogonal projection on a linear subspace $Y \subseteq \rr d$ is $\pi_Y$. 
We denote by $\M_{d_1 \times d_2}( \ro )$ the space of $d_1\times d_2$ matrices with real entries, by $\GL(d,\ro)$ the group of invertible elements of $\M_{d \times d}( \ro )$, and by $\On(d)$ the subgroup of orthogonal matrices in $\GL(d,\ro)$.
The real symplectic group \cite{Folland1} is denoted $\Sp(d,\ro)$ and is defined as the matrices in $\GL(2d,\ro)$ that leaves invariant the 
canonical symplectic form on $T^* \rr d$
\begin{equation*}
\sigma((x,\xi), (x',\xi')) = \la x' , \xi \ra - \la x, \xi' \ra, \quad (x,\xi), (x',\xi') \in T^* \rr d.
\end{equation*}

For a function $f$ on $\rr d$ and $A \in \GL(d,\ro)$ we denote the pullback by $A^* f = f \circ A$. 
The determinant of $A \in \M_{d \times d}( \ro )$ is $|A|$, the transpose is $A^t$, and the inverse of the transpose is $A^{-t}$. 

\subsection*{An integral transform of FBI type} 

\begin{defn}\label{def:FBItransform}
Let $u\in \cS^\prime(\rr d)$ and let $g\in \cS(\rr d)\setminus\{0\}$ be a \textit{window function}. Then the transform $\cT_g u: \rr {2d} \rightarrow \co$ is 
\begin{equation}
\label{eq:cTdef}
\cT_g u(x,\xi)=(2\pi)^{-d/2}(u,T_x M_{\xi}g), \quad x, \xi \in \rr d. 
\end{equation}
\end{defn}

If $u \in \cS(\rr d)$ then $\cT_g u \in \cS(\rr {2d})$ \cite[Theorem~11.2.5]{Grochenig1}. 
The adjoint $\cT_g^*$ is $(\cT_g^* U, f) = (U, \cT_g f)$ for $U \in \cS'(\rr {2d})$ and $f \in \cS(\rr d)$. 
When $U$ is a polynomially bounded measurable function we write
\begin{equation*}
\cT_g^* U(y) = (2\pi)^{-d/2} \int_{\rr {2d}} U(x,\xi) \, T_{x} M_{\xi} g(y) \, \dd x \, \dd \xi 
\end{equation*}
where the integral is defined weakly so that $(\cT_g^* U, f) = (U, \cT_g f)_{L^2}$ for $f \in \cS(\rr d)$. 
\begin{rem}
For $u\in\cS(\rr d)$ we have
\begin{equation*}
\cT_g u(x,\xi)=(2\pi)^{-d/2} \int_{\rr d} u(y) \, e^{-i \la y-x,\xi \ra} \, \overline{g (y-x)} \ \dd y
 = e^{i\langle x, \xi \rangle} \cF(u \, T_x \overline g) (\xi).
\end{equation*}
\end{rem}
The standard, $L^2$-normalized Gaussian window function on $\rr d$ is denoted $\psi_0(x) = \pi^{-d/4}e^{-|x|^2/2}$.

\begin{prop}
\label{prop:Swdchar}
\rm{\cite[Theorem~11.2.3]{Grochenig1}}
Let $u\in\cS'(\rr d)$ and $g \in \cS(\rr d) \setminus 0$. Then $\cT_g u\in C^\infty(\rr {2d})$ and there exists $N \in \no$ that does not depend on $g$ such that 
\begin{equation}
\label{eq:Swdineq}
|\cT_g u(x,\xi)|\lesssim \eabs{(x,\xi)}^{N}, \quad (x,\xi) \in \rr {2d}.
\end{equation}
We have $u\in \cS(\rr d)$ if and only if for any $N \in \no$ 
\begin{equation}
\label{eq:Swineq}
|\cT_g u(x,\xi)|\lesssim \eabs{(x,\xi)}^{-N}, \quad (x,\xi) \in \rr {2d}. 
\end{equation}
\end{prop}

\begin{rem}
(Relation to other integral transforms.)
The transform $\cT_g$ is related to the short-time Fourier transform (cf. \cite{Grochenig1})
\begin{equation*}
\cV_g u(x,\xi) = (2\pi)^{-d/2}(u,M_{\xi} T_x g), \quad x, \xi \in \rr d, 
\end{equation*}
(for the Gaussian window $g=\psi_0$ also known as the Gabor transform) via
\begin{equation*}
\cT_g u(x,\xi) = e^{i \la x, \xi \ra} \cV_{g}u(x,\xi).
\end{equation*}

For the standard Gaussian window \eqref{eq:cTdef} may be expressed as
\begin{equation}
\label{eq:cTpconvdef}
\cTp u(x,\xi)=(2\pi)^{-d/2}e^{-\frac{|\xi|^2}{2}} ( u * \psi_0 )(x-i\xi)=\mathcal{B}  u(x-i\xi) \, e^{-(|x|^2 + |\xi|^2)/2},
\end{equation}
where $\mathcal{B}$ stands for the Bargmann transform \cite{Grochenig1}. 
\end{rem}

We have for two different windows $g,h\in \cS(\rr d)$ 
\begin{equation}
\label{eq:reproducing}
\cT_h^*\cT_g u=(h,g) u, \qquad u \in \cS'(\rr d), 
\end{equation}
and consequently, $\|g\|_{L^2}^{-2}\cT_g^*\cT_g u=u$ for $g \in \cS(\rr d) \setminus 0$ \cite{Grochenig1}.
If $(h,g) \neq 0$ the inversion formula \eqref{eq:reproducing} can  be written as
\begin{equation*}
( u,f) = (h,g)^{-1}  (\cT_g u, \cT_h f), \qquad u \in \cS'(\rr d), \quad f \in \cS(\rr d).
\end{equation*}

Two important features of $\cT_g$ which distinguishes it from the short-time Fourier transform are the following differential identities. 
\begin{align}
\label{eq:diffident}
\partial_x^\alpha \cT_g u (x,\xi) & = \cT_g (\partial^\alpha u) (x,\xi), \qquad \alpha \in \nn d, \\
\label{eq:diffidentstar}
D_{\xi}^\beta \cT_g u (x,\xi) & = \cT_{g_\beta} u (x,\xi), \qquad \beta \in \nn d, \qquad g_\beta (x) = (-x)^\beta g(x).
\end{align}

As described in \cite{Grochenig1} for the short time Fourier transform, \eqref{eq:reproducing} may be used to estimate the behavior of $\cT_g$ under a change of window.
The following version of this result takes derivatives into account:
\begin{lem}
\label{lem:windchange}
Let $u \in \cS'(\rr d)$ and let $f,g,h \in\cS(\rr d) \setminus 0$ satisfy $(h,g) \neq 0$. 
Then for all $\alpha,\beta \in \nn d$ and $(x,\xi) \in \rr {2d}$
\begin{equation}\label{eq:Tderivative}
|\partial_x^\alpha \partial_{\xi}^\beta \cT_f u(x,\xi)| \leqs (2\pi)^{-d/2} |(h,g)|^{-1} |\partial_x^\alpha \cT_g u|*|\cT_{f_\beta} h| (x,\xi).
\end{equation}
\end{lem}
\begin{proof}
We obtain from \eqref{eq:reproducing}
\begin{equation}
\label{eq:reprointer}
\cT_f u = (h,g)^{-1} \cT_f \cT_h^*\cT_g u.
\end{equation}
We may express $\cT_f \cT_h^*\cT_g u$ as
\begin{align*}
\cT_f \cT_h^*\cT_g u(x,\xi)
& = (2\pi)^{-d/2} ( \cT_g u, \cT_h (T_x M_\xi f) ) \\
& = (2\pi)^{-d} \int_{\rr {2d}} \cT_g u(y,\eta) \, (T_y M_\eta h, T_x M_\xi f) \, \dd y \, \dd \eta \\
& = (2\pi)^{-d/2} \int_{\rr {2d}} e^{i \la x-y,\eta \ra} \cT_g u(y,\eta) \, \cT_f h(x-y,\xi-\eta) \, \dd y \, \dd \eta \\
& = (2\pi)^{-d/2} \int_{\rr {2d}} e^{i \la y,\eta \ra} \cT_g u(x-y,\eta) \, \cT_f h(y,\xi-\eta) \, \dd y \, \dd \eta.
\end{align*}
Combining \eqref{eq:diffident}, \eqref{eq:diffidentstar} and \eqref{eq:reprointer} yields
\begin{align*}
& \partial_x^\alpha D_{\xi}^\beta \cT_f u (x,\xi) \\
& = (2\pi)^{-d/2} (h,g)^{-1} \int_{\rr {2d}} 
e^{i \la x-y,\eta \ra} \partial_y^\alpha\cT_g u(y,\eta) \, \cT_{f_\beta} h(x-y,\xi-\eta) \, \dd y\, \dd \eta.
\end{align*}
Taking absolute value gives \eqref{eq:Tderivative}.
\end{proof}

\subsection{Transformation under shifts and symplectic matrices}

A pseudodifferential operator in the Weyl quantization is for $f \in \cS(\rr d)$ defined as
\begin{equation} \label{shubop}
a^w(x,D) f(x) = \int_{\rr {2d}} e^{i \la x-y, \xi \ra} a\left((x+y)/2,\xi \right) \, f(y) \, \dbar \xi \, \dd y 
\end{equation}
where $a$ is the Weyl symbol. 
We will later use Shubin symbols, but for now it suffices to note that the Weyl quantization extends 
by the Schwartz kernel theorem to $a \in \cS'(\rr {2d})$, and then gives 
rise to a continuous linear operator from $\cS(\rr d)$ to $\cS'(\rr d)$. 

The Schwartz kernel of the operator $a^w(x,D)$ is
\begin{equation}\label{eq:schwartzkernelpseudo}
K_{a}(x,y)=\int_{\rr d} e^{i \la x-y, \xi \ra} a\left((x+y)/2,\xi \right) \dbar \xi
\end{equation}
interpreted as a partial inverse Fourier transform in $\cS'(\rr {2d})$ when $a \in \cS'(\rr {2d})$. 

The \emph{metaplectic representation} \cite{Folland1,Taylor1} works as follows.
To each symplectic matrix $\chi \in \Sp(d,\ro)$ is associated an operator $\mu(\chi)$ that is unitary on $L^2(\rr d)$, and determined up to a complex factor of modulus one, such that
\begin{equation}\label{metaplecticoperator}
\mu(\chi)^{-1} a^w(x,D) \, \mu(\chi) = (a \circ \chi)^w(x,D), \quad a \in \cS'(\rr {2d})
\end{equation}
(cf. \cite{Folland1,Hormander0}).
The operator $\mu(\chi)$ is a homeomorphism on $\mathscr S$ and on $\mathscr S'$.

The metaplectic representation is the mapping $\Sp(d,\ro) \ni \chi \rightarrow \mu(\chi)$.
It is in fact a representation of the so called $2$-fold covering group of $\Sp(d,\ro)$, which is called the metaplectic group. 

In Table \ref{tab:meta} we list the generators $\chi$ of the symplectic group, the corresponding unitary operators $\mu(\chi)$ on $u \in L^2$, 
and the corresponding transformation on $\cT_g u$, cf. \cite{DeGosson}. 
We also list the correspondence for phase shift operators. 
Here $x_0, \xi_0 \in \rr d$, $A \in \GL(d,\ro)$, $B \in M_{d\times d}(\mathbb{R})$ with $B=B^t$.

 \renewcommand{\arraystretch}{1.2}
\begin{table}[htb!]
\caption{The metaplectic representation}
\label{tab:meta}
\begin{tabular}{|c||c|}
\hline
\textbf{Transformation} & \textbf{Action on:} \\
 & $(x,\xi)\in T^* \rr d$ \\ 
 &  $u\in L^2(\rr{d})$ \\ 
 & $\cT_g u(x,\xi) \in L^2(\rr {2d})$ \\
\hline \hline
& $(A^{-1}x,A^t\xi)$ \\ 
Coordinate change &  $|A|^{1/2} A^* u$ \\ 
& $|A|^{-1/2} \cT_{A^{-*}g} u (A x,A^{-t} \xi)$\\ 
\hline 
& $(\xi,-x)$ \\ 
Rotation $\pi/2$ & $\cF u$ \\ 
& $e^{i \la x,\xi \ra} \cT_{\cF^{-1} g} u (-\xi,x)$ \\ 
\hline 
& $(x,\xi+Bx)$ \\ 
Shearing & $e^{\frac{i}{2} \la x,Bx \ra}u(x)$ \\ 
& $e^{\frac{i}{2} \la x,B x \ra}\cT_{g_{B}} u(x,\xi-Bx)$ \\ 
\hline 
& $(x+x_0,\xi+\xi_0)$ \\ 
Shift & $T_{x_0} M_{\xi_0} u$ \\ 
& $ e^{i \la \xi_0,x-x_0 \ra} \cT_g u (x-x_0,\xi-\xi_0)$ \\ 
\hline
\end{tabular} 
\end{table}
 \renewcommand{\arraystretch}{1}

The proofs of the claims in Table \ref{tab:meta} are collected in the following lemmas. 

\begin{lem}
\label{lem:shift}
Let $u \in \cS '(\rr d)$ and $g \in  \cS(\rr d) \setminus 0$. 
If $(x_0,\xi_0) \in T^* \rr d$, $A\in\GL(d,\ro)$, $B \in \M_{d \times d}(\ro)$ is symmetric, $v(x)=e^{\frac{i}{2} \la x,Bx \ra} u(x)$
and $g_B(y)= e^{-\frac{i}{2} \la y, B y \ra}g(y)$, 
then for $(x,\xi) \in T^* \rr d$
\begin{align*}
\cT_g(T_{x_0} M_{\xi_o} u)(x,\xi) & = e^{i \la \xi_0, x-x_0 \ra} \cT_g u(x-x_0,\xi-\xi_0), \\
\cT_g( |A|^{1/2} A^*u)(x,\xi) & = |A|^{-1/2}\cT_{A^{-*}g} u (A x,A^{-t} \xi), \\
\cT_g v(x,\xi) & = e^{\frac{i}{2} \la x,Bx \ra }\cT_{g_B} u(x,\xi - Bx). 
\end{align*}
\end{lem}

\begin{proof}
The first and the fourth entry of Table \ref{tab:meta} are immediate consequences of Definition \ref{def:FBItransform}.
For the third identity, assume first $u\in\cS(\rr d)$. Then
\begin{align*}
\cT_g v (x,\xi) & = (2\pi)^{-d/2} \int_{\rr d} \overline{g(y-x)} \, e^{\frac{i}{2} \la y, B y \ra} u(y) e^{i \la x-y,\xi \ra}\ \dd y \\
&=(2\pi)^{-d/2} e^{\frac{i}{2} \la x,B x \ra} \int_{\rr d} \overline{g(y-x) \, e^{-\frac{i}{2} \la y-x, B (y-x) \ra }} u(y) e^{i \la x-y, \xi-Bx \ra }\ \dd y \\
&=e^{\frac{i}{2} \la x, B x \ra } \cT_{g_B} u(x,\xi-Bx).
\end{align*}
The formula extends to $u \in \cS'(\rr d)$.
\end{proof}

Finally we prove the claim for ``Rotation $\pi/2$'' in Table \ref{tab:meta}. 
For later use, we prefer to show a more general result for a possibly partial Fourier transform.

\begin{lem}
\label{lem:Fourier}
If $u \in \cS'(\rr d)$, $0 \leqs n \leqs d$ and $x=(x_1,x_2) \in \rr d$, $x_1 \in \rr n$, $x_2 \in \rr {d-n}$, 
then 
\begin{equation*}
\cT_g  u(x_1,x_2,\xi_1,\xi_2) = e^{i \la x_2,\xi_2 \ra} \cT_{\cF_2 g} \cF_2  u (x_1,\xi_2,\xi_1, -x_2).
\end{equation*}
\end{lem}
\begin{proof}
\begin{multline*}
e^{i \la x_2,\xi_2 \ra} \cT_{\cF_2 g} \cF_2  u (x_1,\xi_2,\xi_1, -x_2) \\
= e^{i \la x_2,\xi_2 \ra} (2\pi)^{-d/2} ( \cF_2  u, T_{x_1,\xi_2} M_{\xi_1,-x_2} \cF_2 g) \\
= e^{i \la x_2,\xi_2 \ra} (2\pi)^{-d/2} ( u, \cF_2^{-1} T_{x_1,\xi_2} M_{\xi_1,-x_2} \cF_2 g) \\ = (2\pi)^{-d/2} ( u, T_{x_1,x_2} M_{\xi_1,\xi_2} g). 
\end{multline*}
\end{proof}

\begin{rem}
The extreme cases $n=0$ and $n=d$ represent $\cF_2 = \cF$ (the full Fourier transform) and the trivial case $\cF_2 = I$ (the identity), respectively. 
\end{rem}

We observe that up to certain phase factors, changes of windows and sign conventions, the ``Action on $\cT_g u(x,\xi)$'' reflects the inversion of ``Action on $T^* \rr d$'' in Table \ref{tab:meta}.

%%%%%%%%%%%%%%%%%%%%%%
\section{Characterization of Shubin symbols}
%%%%%%%%%%%%%%%%%%%%%%
\label{sec:shubchar}
We first recall the definition of Shubin's class of global symbols for pseudodifferential operators \cite{Shubin1}.
\begin{defn}
We say that $a \in C^\infty(\rr d)$ is a Shubin symbol of order $m \in \ro$ and parameter $0 \leqs \rho \leqs 1$, denoted $a \in \Gamma_{\rho}^m(\rr d)$, if there exist $C_\alpha>0$ such that
\begin{equation}
\label{eq:shubinineq}
|\partial^\alpha a(z)| \leqs C_\alpha \eabs{z}^{m-\rho|\alpha|}, \qquad \alpha \in \nn d, \quad z \in \rr d.
\end{equation}
$\Gamma_\rho^m(\rr d)$ is a  Fr\'echet space equipped with the seminorms $\rho^m_M(a)$ of best possible constants $C_\alpha$ in \eqref{eq:shubinineq} maximized over $|\alpha| \leqs M$, $M \in \no$.
We denote $\Gamma^m(\rr d) = \Gamma_1^m(\rr d)$.
\end{defn}
Obviously $\Gamma_\rho^m(\rr d) \subseteq \cS'(\rr d)$ so Proposition \ref{prop:Swdchar} already gives some information on $\cT_g a$ when $a\in \Gamma_\rho^m(\rr d)$. 
The following result, which is a chief tool in the paper, gives characterizations of $\cT_g a$ for $a \in \Gamma_\rho^m(\rr d)$. 
\begin{prop}
\label{prop:symbchar}
Suppose $a\in \cS'(\rr d)$. 
Then $a\in \Gamma_\rho^m(\rr d)$ if and only if for one (and equivalently all) $g\in\cS(\rr d)\setminus 0$
\begin{equation}
\label{eq:Gineq1}
|\partial_x^\alpha \partial_\xi^\beta \cT_g a(x,\xi)|\lesssim \eabs{x}^{m- \rho |\alpha|}\eabs{\xi}^{-N}, \quad N \geqs 0, \quad \alpha, \beta \in \nn d, \quad x, \xi \in \rr d, 
\end{equation}
or equivalently
\begin{equation}
\label{eq:Gineq}
|\partial_x^\alpha \cT_g a(x,\xi)|\lesssim \eabs{x}^{m- \rho |\alpha|}\eabs{\xi}^{-N}, \quad N \geqs 0, \quad \alpha \in \nn d, \quad x, \xi \in \rr d. 
\end{equation}
\end{prop}
\begin{proof}
Let $a\in \Gamma_\rho^m(\rr d)$, let $g\in\cS(\rr d)\setminus 0$ and let $\alpha,\beta, \gamma \in\nn d$ be arbitrary. We seek to show
$$
|\xi^\gamma \partial_x^\alpha \partial_\xi^\beta \cT_g a(x,\xi)| \lesssim \eabs{x}^{m-\rho|\alpha|}. 
$$
To that end we use \eqref{eq:diffident} and \eqref{eq:diffidentstar}, integrate by parts and estimate using \eqref{eq:Peetre} and the fact that $g \in \cS$
\begin{align*}
|\xi^\gamma \partial_x^\alpha \partial_\xi^\beta \cT_g a(x,\xi)|
& = \left| \xi^\gamma \cT_{g_\beta}(\partial^\alpha a)(x,\xi) \right|  \\
& = (2 \pi)^{-d/2} \left| \int_{\rr d} \left((i\partial_{y})^\gamma e^{- i \la \xi,y \ra }\right) \overline{g_\beta(y)} \, \partial^\alpha a(x+y)\ \dd y \right| \\
&\lesssim \int_{\rr{d}} \left| \partial_{y}^\gamma \left[\overline{g_\beta(y)} \, \partial^\alpha a(x+y)\right] \right|\ \dd y  \\
& = \int_{\rr{d}} \left| \sum_{\kappa \leqs \gamma} \binom{\gamma}{\kappa} \partial^{\gamma-\kappa}\overline{g_\beta(y)} \, \partial^{\alpha+\kappa} a(x+y)\right|\ \dd y  \\
& \lesssim \sum_{\kappa \leqs \gamma} \binom{\gamma}{\kappa}  \int_{\rr d} \left| \partial^{\gamma-\kappa} g_\beta(y) \right| \, \eabs{x+y}^{m-\rho|\alpha+\kappa|} \dd y  \\
& \lesssim \eabs{x}^{m-\rho|\alpha|} \sum_{\kappa \leqs \gamma} \binom{\gamma}{\kappa}  \int_{\rr d} \left| \partial^{\gamma-\kappa} g_\beta(y) \right| \, \eabs{y}^{|m|+\rho |\alpha+\kappa|} \dd y  \\
&\lesssim \eabs{x}^{m- \rho |\alpha|}. 
\end{align*}
This implies \eqref{eq:Gineq1} and as a special case \eqref{eq:Gineq}. 

Conversely, suppose that \eqref{eq:Gineq} holds for $a \in \cS'(\rr d)$ for some $g \in \cS(\rr d) \setminus 0$, which is a weaker assumption than \eqref{eq:Gineq1}. 
We obtain from \eqref{eq:reproducing} that $a$ is given by
\begin{align*}
a(y)
& = \|g\|_{L^2}^{-2} \, \cT_g^* \cT_g a(y) \\
& = \|g\|_{L^2}^{-2} \, (2 \pi)^{-d/2} \int_{\rr {2d}} \cT_g a(x,\xi) \, e^{i \la \xi,y-x \ra} \, g(y-x) \, \dd x \, \dd \xi
\end{align*}
which is an absolutely convergent integral due to \eqref{eq:Gineq} and the fact that $g\in\cS(\rr d)$. 
We may differentiate under the integral, so integration by parts, \eqref{eq:Gineq} and \eqref{eq:Peetre} give for any $\alpha \in \nn d$ and any $y \in \rr d$
\begin{align*}
\left|\partial^\alpha a(y)\right|
& = \|g\|_{L^2}^{-2} \, (2 \pi)^{-d/2}  \left| \int_{\rr {2d}} \cT_ga(x,\xi) \, \partial_y^\alpha \left( e^{i \la \xi,y-x \ra} \, g(y-x) \right) \, \dd x \, \dd \xi \right| \\
& = \|g\|_{L^2}^{-2} \, (2 \pi)^{-d/2}  \left| \int_{\rr {2d}} \cT_ga(x,\xi) \, (-\partial_x)^\alpha \left( e^{i \la \xi,y-x \ra} \, g(y-x) \right) \, \dd x \, \dd \xi \right| \\
& = \|g\|_{L^2}^{-2} \, (2 \pi)^{-d/2}  \left| \int_{\rr {2d}} \partial_x^\alpha \cT_ga(x,\xi) \, e^{i \la \xi,y-x \ra} \, g(y-x) \, \dd x \, \dd \xi \right| \\
& = \|g\|_{L^2}^{-2} \, (2 \pi)^{-d/2}  \left| \int_{\rr {2d}} \partial_x^\alpha \cT_ga(y-x,\xi) \, e^{i \la \xi,x \ra} \, g(x) \, \dd x \, \dd \xi \right| \\
& \lesssim \int_{\rr {2d}}  \eabs{y-x}^{m- \rho |\alpha|} \, \eabs{\xi}^{-d-1} \, |g(x)| \, \dd x \, \dd \xi \\
& \lesssim \eabs{y}^{m- \rho |\alpha|} \int_{\rr {2d}}  \eabs{\xi}^{-d-1} \, \eabs{x}^{|m|+ \rho |\alpha|} \, |g(x)| \, \dd x \, \dd \xi \\
& \lesssim \eabs{y}^{m- \rho |\alpha|}.  
\end{align*}
Thus $a \in \Gamma_\rho^m(\rr d)$. 
\end{proof}
\begin{rem}
It follows from the proof that the best possible constants in \eqref{eq:Gineq} maximized over $|\alpha| \leqs M$ yield seminorms $\rho^m_{g,M,N}$,  $M,N \in \no$, on $\Gamma_\rho^m(\rr{d})$ equivalent to $\rho^m_{M}$, $M \in \no$.
\end{rem}

We will next reformulate the characterization of $\Gamma^m(\rr d)$ in a more geometric form.

\begin{prop}
\label{prop:symbchargeom}
Let $a\in \cS'(\rr d)$. 
Then $a \in \Gamma^m(\rr d)$ if and only if for one (and equivalently all) $g \in \cS(\rr d) \setminus 0$ and all $N,k \in \no$
\begin{equation}
\label{eq:Gineqgeom}
|L_1 \cdots L_k \cT_g a(x,\xi)|\lesssim \eabs{x}^m\eabs{\xi}^{-N}, \quad (x,\xi) \in T^* \rr d, 
\end{equation}
for any vector fields of the form $L_i = x_j\partial_{x_n}$ where $1 \leqs j,n \leqs d$, $i=1,\dots,k$. 
\end{prop}

\begin{proof}
We may write 
\begin{equation*}
L_1 \cdots L_k = \sum_{|\alpha|=|\beta| \leqs k} c_{\alpha\beta} x^\alpha \partial^\beta, \quad c_{\alpha\beta} \in \ro, 
\end{equation*}
and all differential operators of this form are linear combinations of products of the vector fields $L_i$. 

If $a \in \Gamma^m(\rr d)$ then the estimates \eqref{eq:Gineq} hold for any $g \in \cS(\rr d) \setminus 0$. 
For $N,k \in \no$ we have
\begin{align*}
|L_1 \cdots L_k \cT_g a(x,\xi)| & \lesssim \sum_{|\alpha|=|\beta| \leqs k} \eabs{x}^{|\alpha|} |\partial_x^\beta \cT_g a(x,\xi)| \\
&\lesssim \eabs{x}^m \eabs{\xi}^{-N}
\end{align*}
which confirms \eqref{eq:Gineqgeom}. 

Suppose on the other hand that the estimates \eqref{eq:Gineqgeom} hold for some $g \in \cS(\rr d) \setminus 0$ and $N,k \in \no$. 
Then for any $\alpha,\beta \in \nn d$ such that $|\alpha|=|\beta|$ and $N \in \no$
\begin{equation*}
|x^\alpha \partial_x^\beta \cT_g a(x,\xi)| \lesssim \eabs{x}^m \eabs{\xi}^{-N}. 
\end{equation*}
This gives using $|x|^{|\beta|}\leqs d^{|\beta|/2} \max_{|\alpha|=|\beta|} |x^\alpha|$ 
\begin{equation*}
|\partial_x^\beta \cT_g a(x,\xi)| \lesssim \eabs{x}^{m-|\beta|} \eabs{\xi}^{-N}, \quad |x| > 1, \quad \xi \in \rr d. 
\end{equation*}
In order to prove \eqref{eq:Gineq}, which is equivalent to $a \in \Gamma^m(\rr d)$,
it thus remains to show that $\eabs{\xi}^N | \partial_{x}^\beta \cT_g a (x,\xi)|$ remains uniformly bounded for $|x| \leqs 1$ and $ \xi \in \rr d$, for any $N \in \nn{}$. For that we estimate
\begin{align*}
\eabs{\xi}^N | \partial_{x}^\beta \cT_g a(x,\xi)|
& = \eabs{\xi}^N \left| \sum_{\alpha \leqs \beta} c_{\alpha\beta} (i\xi)^{\alpha} \cT_{\partial^{\beta-\alpha} g} a (x,\xi) \right|\\
& \lesssim \eabs{\xi}^{|\beta|+N} \sum_{\alpha \leqs \beta} \left| \cT_{\partial^{\alpha} g} a (x,\xi) \right|.
\end{align*}
By Lemma \ref{lem:windchange} we have
\begin{equation*}
\left| \cT_{\partial^{\alpha} g} a (x,\xi) \right|
\lesssim \big( \underbrace{|\cT_g a|}_{\lesssim \eabs{x}^{m}\eabs{\xi}^{-M}} * \ |\underbrace{\cT_{\partial^{\alpha} g} g}_{\in \cS}| \big) (x,\xi)\lesssim \eabs{x}^{m}\eabs{\xi}^{-M}
\end{equation*}
where the last inequality follows by Peetre's inequality \eqref{eq:Peetre} applied to the convolution. Choosing $M \geqs |\beta|+N$, we obtain
\begin{equation*}
\eabs{\xi}^N \left|\partial_{x}^\beta \cT_g a(x,\xi)\right|\lesssim 1\qquad \text{ for } |x| \leqs 1, \quad \xi \in \rr d,
\end{equation*}
which proves the claim.
\end{proof}

\begin{rem}
The vector fields $x_j\partial_{x_n}$ play a role in spanning all vector fields tangential to $\{0 \} \times \rr d \subseteq T^* \rr d$, see \cite[Lemma 18.2.5]{Hormander0}. 
\end{rem}

\subsection{Classical symbols}
An important subclass of the Shubin symbols are those that admit a polyhomogeneous expansion, so called classical symbols. 
A symbol $a\in \Gamma^m(\rr {d})$ is called classical, denoted $a \in \Gamma^m_\cl(\rr {d})$, if there are functions $a_{m-j}$, homogeneous of degree $m-j$ and smooth outside $z=0$, $j=0,1,\dots$, such that for any zero-excision function\footnote{This means a function of the form $1-\phi$ where $\phi\in C_c^\infty(\rr {d})$ and $\phi\equiv 1$ near zero.} $\chi$ we have for any $N \in \no$ 
\begin{equation*}
a-\chi\sum_{j=0}^{N-1} a_{m-j}\in \Gamma^{m-N}(\rr {d}).
\end{equation*}
By Euler's relation for homogeneous functions, $u$ is homogeneous of degree $m$ if and only if $R u = m u,$ where $R$ is the radial vector field $Ra(x)=\la x, \nabla a(x) \ra$. Adapting the method of Joshi \cite{Joshi} gives the following characterization of classical Shubin symbols.
\begin{prop}
\label{lem:class}
A symbol $a\in \Gamma^m(\rr {d})$ is classical if and only if for all $N \in \no_0$
\begin{equation*}
(R-m+N-1) (R-m+N-2) \cdots (R-m) \, a  \in \Gamma^{m-N}(\rr {d}).
\end{equation*}
\end{prop}
The transformation $a \rightarrow \cT_g a$ does not preserve homogeneity. Nevertheless
\eqref{eq:diffident} and \eqref{eq:diffidentstar} give the relation
\begin{equation*}
\cT_g \left(R a \right) (x,\xi) = \la x + i \nabla_\xi, \nabla_x \ra \cT_g a(x,\xi) =: \wt{R} \cT_g a(x,\xi).
\end{equation*}
\begin{cor}
\label{cor:classsymbchar}
Let $a \in \cS'(\rr {d})$ and $g \in \cS(\rr {d}) \setminus 0$. Then $a \in \Gamma_\cl^m(\rr {d})$ if and only if
\begin{multline}
\label{eq:classtransf}
\left| \partial_x^\alpha \left( (\wt{R}-m+N-1) (\wt{R}-m+N-2) \cdots (\wt{R}-m) \cT_g a(x,\xi) \right) \right| \\
\lesssim \eabs{x}^{m- N-|\alpha|}\eabs{\xi}^{-M}
\end{multline}
for any $M \geqs 0$, $N\in \no_0$, $\alpha \in \nn d$ and $(x,\xi) \in T^* \rr d$.
\end{cor}

\begin{proof}
By Proposition \ref{lem:class}, $a \in \Gamma_\cl^m(\rr {d})$ if and only if 
\begin{equation*}
(R-m+N-1) (R-m+N-2) \cdots (R-m) \, a  \in \Gamma^{m-N}(\rr {d}).
\end{equation*}
By Proposition \ref{prop:symbchar} this holds if and only if for all $\alpha \in \nn d$, $(x, \xi) \in \rr {2d}$, and $M \geqs 0$
\begin{align*}
& |\partial_x^\alpha \cT_g \left( (R-m+N-1) (R-m+N-2) \cdots (R-m) \, a\right)(x,\xi)| \\
& \lesssim \eabs{x}^{m- N-|\alpha|}\eabs{\xi}^{-M}.
\end{align*}
This is equivalent to \eqref{eq:classtransf}.
\end{proof}

%%%%%%%%%%%%%%%%%%%%%%%%%%%%%
\section{Characterization of pseudodifferential operators}
%%%%%%%%%%%%%%%%%%%%%%%%%%%%%
\label{sec:pseudochar}

When $a \in \Gamma_\rho^m(\rr {2d})$ the pseudodifferential operator $a^w(x,D)$ is continuous on $\cS(\rr d)$, and extends to a continuous operator on $\cS'(\rr d)$ \cite{Shubin1}. 
The formulas \eqref{shubop} and \eqref{eq:schwartzkernelpseudo} can be interpreted as oscillatory integrals if $0 < \rho \leqs 1$. 

\begin{lem}
\label{lem:tTpsdo}
Let $a\in \Gamma_\rho^m(\rr {2d})$ and $g \in \cS(\rr {2d}) \setminus 0$. Then, for $(z,\zeta)=( z_1,z_2; \zeta_1,\zeta_2) \in T^* \rr {2d}$,
\begin{multline}
\label{eq:tTpsdo}
\cT_g K_a(z,\zeta)
\\ = (2 \pi)^{-d/2} \cT_h a \left(\frac{z_1+z_2}{2},\frac{\zeta_1-\zeta_2}{2},\zeta_1+\zeta_2,z_2-z_1 \right) e^{\frac{i}{2} \la \zeta_1-\zeta_2,z_1-z_2 \ra}
\end{multline}
where $h = \cF_2 (g \circ \kappa)$, $\kappa(x,y) = (x+y/2,x-y/2)$ and $x,y \in \rr d$. 
\end{lem}

\begin{proof}
The statement \eqref{eq:tTpsdo} can be rephrased as 
\begin{multline*}
\cT_g K_a \left( z_1 - \frac{z_2}{2} , z_1 + \frac{z_2}{2}; \zeta_1 + \frac{\zeta_2}{2},- \zeta_1 + \frac{\zeta_2}{2} \right) \\
= (2 \pi)^{-d/2} \cT_h a \left( z_1, \zeta_1; \zeta_2, z_2 \right) e^{- i \la \zeta_1, z_2 \ra}, 
\end{multline*}
for all $( z_1,z_2; \zeta_1,\zeta_2) \in T^* \rr {2d}$. 
We have $K_a = (2 \pi)^{-d/2} (\cF_2^{-1} a) \circ \kappa^{-1}$ 
which gives 
\begin{equation}\label{eq:STFTSchwartz1}
\begin{aligned}
& \cT_g K_a \left( z_1 - \frac{z_2}{2} , z_1 + \frac{z_2}{2}; \zeta_1 + \frac{\zeta_2}{2},- \zeta_1 + \frac{\zeta_2}{2} \right) \\
& = (2 \pi)^{-3d/2} ( (\cF_2^{-1} a) \circ \kappa^{-1}, T_{z_1 - \frac{z_2}{2} , z_1 + \frac{z_2}{2}} M_{\zeta_1 + \frac{\zeta_2}{2},- \zeta_1 + \frac{\zeta_2}{2}} g) \\
& = (2 \pi)^{-3d/2} ( a, \cF_2( T_{z_1 - \frac{z_2}{2} , z_1 + \frac{z_2}{2}} M_{\zeta_1 + \frac{\zeta_2}{2},- \zeta_1 + \frac{\zeta_2}{2}} g \circ \kappa)). 
\end{aligned}
\end{equation}

We calculate
\begin{align*}
& (2 \pi)^{d/2} \cF_2( T_{z_1 - \frac{z_2}{2} , z_1 + \frac{z_2}{2}} M_{\zeta_1 + \frac{\zeta_2}{2},- \zeta_1 + \frac{\zeta_2}{2}} g \circ \kappa) (y,\eta) \\
& = \int_{\rr d} T_{z_1 - \frac{z_2}{2} , z_1 + \frac{z_2}{2}} M_{\zeta_1 + \frac{\zeta_2}{2},- \zeta_1 + \frac{\zeta_2}{2}} g \circ \kappa(y,u) e^{-i \la u, \eta \ra} \, \dd u \\
& = \int_{\rr d} e^{i \left( \la \zeta_1 + \frac{\zeta_2}{2}, y + \frac{u}{2} - z_1 + \frac{z_2}{2} \ra + \la - \zeta_1 + \frac{\zeta_2}{2}, y - \frac{u}{2} - z_1 - \frac{z_2}{2}  \ra - \la u, \eta \ra \right) } \\
& \qquad \qquad \qquad \qquad \qquad \times g\left( y + \frac{u}{2} - z_1 + \frac{z_2}{2}, y - \frac{u}{2} - z_1 - \frac{z_2}{2} \right) \dd u \\
& = e^{i \left( \la \zeta_1, z_2 \ra + \la \zeta_2, y - z_1 \ra \right) } \int_{\rr d} e^{-i \la u, \eta-\zeta_1 \ra } g\left( y + \frac{u}{2} - z_1 + \frac{z_2}{2}, y - \frac{u}{2} - z_1 - \frac{z_2}{2} \right) \dd u \\
& = e^{i \la \zeta_1, z_2 \ra} e^{i \la \zeta_2, y - z_1 \ra } \int_{\rr d} e^{-i \la u-z_2, \eta-\zeta_1 \ra } g\left( y - z_1 + \frac{u}{2}, y - z_1 - \frac{u}{2} \right) \dd u \\
& = (2 \pi)^{d/2} e^{i \la \zeta_1, z_2 \ra} e^{i \left( \la \zeta_2, y - z_1 \ra + \la z_2, \eta-\zeta_1 \ra \right)} \cF_2 (g \circ \kappa)(y-z_1,\eta-\zeta_1) \\
& = (2 \pi)^{d/2} e^{i \la \zeta_1, z_2 \ra} T_{z_1,\zeta_1} M_{\zeta_2,z_2} \cF_2 (g \circ \kappa)(y,\eta). 
\end{align*}
Insertion into \eqref{eq:STFTSchwartz1} gives the claimed conclusion. 
\end{proof}

\begin{defn}
For $u\in\cS'(\rr {2d})$ and $g \in \cS(\rr {2d}) \setminus 0$ we denote  
\begin{equation*}
\cT_g^\Delta u(z_1,z_2,\zeta_1,\zeta_2) = e^{-\frac{i}{2} \la \zeta_1-\zeta_2, z_1-z_2 \ra }\cT_g u (z_1,z_2,\zeta_1,\zeta_2)
\end{equation*}
for $\quad ( z_1,z_2; \zeta_1,\zeta_2) \in T^* \rr {2d}.$
\end{defn}

As a consequence of Proposition \ref{prop:symbchar} we obtain the following characterization of the Schwartz kernels of Weyl quantized Shubin operators.

\begin{prop}
\label{prop:LGchar}
Let $K\in\cS'(\rr {2d})$. Then $K$ is the Schwartz kernel of an operator of the form \eqref{shubop} with $a \in \Gamma_\rho^m(\rr {2d})$ if and only if for all $\alpha,\beta \in \nn d$ and $N \in \no$ and any $g \in \cS(\rr {2d}) \setminus 0$ we have
\begin{equation}\label{eq:kernelchar1}
\begin{aligned}
& |(\partial_{z_1} + \partial_{z_2})^\alpha (\partial_{\zeta_1} - \partial_{\zeta_2})^\beta \cT_g^\Delta K (z_1,z_2, \zeta_1, \zeta_2)| \\
& \qquad \lesssim \eabs{(z_1+z_2,\zeta_1-\zeta_2)}^{m- \rho |\alpha+\beta|}\eabs{(z_1-z_2,\zeta_1+\zeta_2)}^{-N}, \\
& \qquad \qquad  ( z_1,z_2; \zeta_1,\zeta_2) \in T^* \rr {2d}.
\end{aligned}
\end{equation}
\end{prop}

\begin{rem}\label{rem:kernelchargeom}
Corresponding to Proposition \ref{prop:symbchargeom}, we may rephrase the estimates \eqref{eq:kernelchar1} for $\Gamma^m(\rr {2d})$ as 
\begin{equation*}
|L_1 \cdots L_k \cT_g^\Delta K (z_1,z_2, \zeta_1, \zeta_2)|  
\lesssim \eabs{(z_1+z_2,\zeta_1-\zeta_2)}^{m}\eabs{(z_1-z_2,\zeta_1+\zeta_2)}^{-N}, 
\end{equation*}
where $L_i$ are differential operators of the form
\begin{align*}
L_i & = (z_{1,j} + z_{2,j}) (\partial_{z_{1,n}} + \partial_{z_{2,n}}), 
\quad 
& L_i = (z_{1,j} + z_{2,j}) (\partial_{\zeta_{1,n}} - \partial_{\zeta_{2,n}}), \\
L_i & = (\zeta_{1,j} - \zeta_{2,j}) (\partial_{z_{1,n}} + \partial_{z_{2,n}}), 
\quad \mbox{or} \quad 
& L_i = (\zeta_{1,j} - \zeta_{2,j}) (\partial_{\zeta_{1,n}} - \partial_{\zeta_{2,n}})
\end{align*}
for $1 \leqs j,n \leqs d$ and $1 \leqs i \leqs k$. 
\end{rem}

Proposition \ref{prop:LGchar} may be phrased in terms of the Schwartz kernel $K_{\cT_g a^w(x,D) \cT_h ^*}$ of the operator $\cT_g a^w(x,D) \cT_h^*$ 
for $a \in \Gamma_\rho^m(\rr {2d})$. 
Let $u,v \in \cS(\rr d)$ and $g,h \in \cS(\rr d) \setminus 0$. 
On the one hand
\begin{align*}
( a^w(x,D) u,v) & = \| g \|^{-2}_{L^2} \| h \|^{-2}_{L^2} (\cT_g a^w(x,D) \cT_h^* (\cT_h u), \cT_g v ) \\
& = \| g \|^{-2}_{L^2} \| h \|^{-2}_{L^2} ( K_{\cT_g a^w(x,D) \cT_h^*}, \cT_g v \otimes \overline{\cT_h u} ) 
\end{align*}
and on the other hand
\begin{equation*}
( a^w(x,D) u,v) = ( K_a, v \otimes \overline u ) = \| g \|^{-2}_{L^2} \| h \|^{-2}_{L^2} ( \cT_{g\otimes \overline{h}} K_a,\cT_{g\otimes \overline{h}}(v \otimes \overline u)).
\end{equation*}
Since 
\begin{equation*}
(\cT_g v \otimes \overline{\cT_h u}) (z_1,\zeta_1, z_2, \zeta_2) = \cT_{g\otimes\overline{h}} (v \otimes \overline u) (z_1,z_2,\zeta_1,-\zeta_2)
\end{equation*}
this proves the formula
\begin{equation}
\label{eq:cTpkernelident}
K_{\cT_g a^w(x,D) \cT_h^*} (z_1,\zeta_1,z_2,-\zeta_2 )=\cT_{g\otimes\overline{h}} K_a (z_1,z_2,\zeta_1,\zeta_2).
\end{equation}

In view of the last identity and Proposition \ref{prop:LGchar} we have the following result. 
Tataru \cite[Theorem 1]{Tataru} obtained a version of this characterization in the special case $\Gamma_0^0$, and $\alpha =\beta =0$. 

\begin{cor}
We have $a \in \Gamma_\rho^m(\rr {2d})$ if and only if for all $\alpha,\beta \in \nn d$ and $N \in \no$ and any $g,h \in \cS(\rr {2d}) \setminus 0$
\begin{multline}
\left| (\partial_{z_1} + \partial_{z_2})^\alpha (\partial_{\zeta_1} - \partial_{\zeta_2})^\beta \left( e^{- \frac{i}{2} \la z_1-z_2, \zeta_1-\zeta_2 \ra} K_{\cT_g a^w(x,D) \cT_h^*} (z_1,\zeta_1,z_2,-\zeta_2 )\right) \right| \\  \lesssim  \eabs{(z_1+z_2,\zeta_1-\zeta_2)}^{m-\rho |\alpha+\beta|}\eabs{(z_1-z_2,\zeta_1+\zeta_2)}^{-N},  \quad  ( z_1,z_2; \zeta_1,\zeta_2) \in T^* \rr {2d}.
\end{multline}
\end{cor}

\subsection{Continuity in Shubin--Sobolev spaces}

As an application of the previous characterization we give a simple proof of continuity of Shubin pseudodifferential operators in isotropic Sobolev spaces. The Shubin--Sobolev spaces $Q^s(\rr d)$, $s \in \ro$, introduced by Shubin \cite{Shubin1} (cf. \cite{Grochenig1,Nicola1}) can be defined as the modulation space $M^{2}_s(\rr d)$, that is 
\begin{equation*}
Q^s (\rr d)  = \{u \in \cS'(\rr d): \, \eabs{ \cdot }^s \cT_g u \in L^2(\rr {2d}) \} 
\end{equation*}
where $g \in \cS(\rr d) \setminus 0$ is fixed and arbitrary, with norm
\begin{equation*}
\| u \|_{Q^s} = \left\| \eabs{ \cdot }^s \cT_g u \right\|_{L^2(\rr {2d})}. 
\end{equation*}

The characterization of Shubin pseudodifferential operators given in Proposition \ref{prop:LGchar} yields a simple proof of their $Q^s$-continuity, cf. \cite{Tataru}.
\begin{prop}
If $a \in \Gamma_0^m(\rr {2d})$ then $a^w(x,D) :Q^{s+m}(\rr d)\rightarrow Q^s(\rr d)$ is continuous for all $s \in \ro$.
\end{prop}
\begin{proof}
Set $A=a^w(x,D)$. 
We have for $u\in Q^{s+m}(\rr{d})$  
\begin{align*}
\|A u\|_{Q^s}
& = \sup_{v\in Q^{-s}} |( A u,v)| 
= \sup_{v\in Q^{-s}} |( K_{\cTp A \cTp^*}, \cTp v \otimes \overline{\cTp u} )| \\
& =  \sup_{v\in Q^{-s}} |( \eabs{\cdot}^{s} \otimes \eabs{\cdot}^{-s-m} K_{\cTp A \cTp^*}, \underbrace{\eabs{\cdot}^{-s} \, \cTp v}_{\in L^2(\rr {2d})} \otimes \underbrace{\eabs{\cdot}^{s+m} \, \overline{ \cTp u}}_{\in L^2(\rr {2d})} )|. 
\end{align*}
It remains to show that 
\begin{equation}\label{eq:schwartzkernel1}
\eabs{(z_1,\zeta_1)}^{s}\eabs{(z_2,\zeta_2)}^{-s-m} K_{\cTp A\cTp^*} (z_1,\zeta_1, z_2, \zeta_2) 
\end{equation}
is the Schwartz kernel of a continuous operator on $L^2(\rr {2d})$. 

First we deduce from \eqref{eq:cTpkernelident}, Proposition \ref{prop:LGchar} and \eqref{eq:Peetre} the estimate for any $N \in \no$
\begin{align*}
& \eabs{(z_1,\zeta_1)}^{s}\eabs{(z_2,\zeta_2)}^{-s-m} |K_{\cTp A\cTp^*} (z_1,\zeta_1, z_2, \zeta_2)| \\
& = \eabs{(z_1,\zeta_1)}^{s}\eabs{(z_2,\zeta_2)}^{-s-m} |\cTp K_a (z_1,z_2, \zeta_1, -\zeta_2)| \\
& \lesssim \eabs{(z_1,\zeta_1)}^{s}\eabs{(z_2,\zeta_2)}^{-s-m} \eabs{(z_1+z_2, \zeta_1+\zeta_2)}^m \eabs{(z_1-z_2,\zeta_1-\zeta_2)}^{-N} \\
& \lesssim \eabs{(z_2,\zeta_2)}^{-m} \eabs{(z_1, \zeta_1)+(z_2,\zeta_2)}^m \eabs{(z_1,\zeta_1)-(z_2,\zeta_2)}^{|s| -N} \\
& \lesssim \eabs{(z_1,\zeta_1)-(z_2,\zeta_2)}^{|s|+|m| - N}.
\end{align*}

Then we apply Schur's test which gives, for $N>0$ sufficiently large,
\begin{align*}
\int_{\rr {2d}} \left|\eabs{(z_1,\zeta_1)}^{s}\eabs{(z_2,\zeta_2)}^{-s-m} K_{\cTp A\cTp^*} (z_1,\zeta_1, z_2, \zeta_2) \right| \, \dd z_1 \, \dd \zeta_1 & \lesssim 1,\\
\int_{\rr {2d}} \left|\eabs{(z_1,\zeta_1)}^{s}\eabs{(z_2,\zeta_2)}^{-s-m} K_{\cTp A\cTp^*} (z_1,\zeta_1, z_2, \zeta_2) \right| \, \dd z_2 \, \dd \zeta_2 & \lesssim 1.
\end{align*}
This implies that \eqref{eq:schwartzkernel1} is the Schwartz kernel of an operator that is continuous on $L^2(\rr {2d})$. 
\end{proof}
%

%%%%%%%%%%%%%%%%%%%%%%
\section{$\Gamma$-conormal distributions}
%%%%%%%%%%%%%%%%%%%%%%
\label{sec:gconorm}

The kernels of pseudodifferential operators with H\"ormander symbols are prototypes of conormal distributions, see \cite[Chapter~18.2]{Hormander0}. We introduce an analogous notion in the Shubin calculus. Before giving a precise definition we make some observations to clarify our idea. 

Proposition \ref{prop:LGchar} may be rephrased using the diagonal and the antidiagonal
\begin{equation*}
\Delta = \{(x,x) \in \rr {2d}: \ x \in \rr d \}, \qquad \Delta^\perp = \{(x,-x) \in \rr {2d}: \ x \in \rr d \} 
\end{equation*}
considered as linear subspaces of $\rr{2d}$.
Denoting Euclidean distance to a subset $V$ by $\dist(\cdot,V)$ we have
\begin{equation*}
\dist((x,y),\Delta) = \inf_{z \in \rr d} \left| (x,y) - (z,z) \right| = \frac{|x-y|}{\sqrt{2}}, \quad (x,y) \in \rr {2d}, 
\end{equation*}
and $\dist((x,y),\Delta^\perp)) = |x+y|/\sqrt{2}$ for $(x,y) \in \rr {2d}$. 

The inequalities \eqref{eq:kernelchar1} can thus be expressed, for $(x,\xi) \in T^* \rr {2d}$, as
\begin{equation}\label{eq:kernelchar2}
\begin{aligned}
\left| L_1 \cdots L_k \cT_g^\Delta K_a (x,\xi) \right |
& \lesssim \left( 1 + \dist((x,\xi),N(\Delta^\perp)) \right)^{m- \rho k} \\ 
& \qquad \times \left( 1 + \dist((x,\xi),N(\Delta)) \right)^{-N},
\end{aligned}
\end{equation}
where $N(\Delta) =\Delta \times  \Delta^\perp \subseteq T^* \rr {2d}$ and $N(\Delta^\perp) = \Delta^\perp \times \Delta \subseteq T^* \rr {2d}$ denote the conormal spaces of $\Delta$ and $\Delta^\perp$ respectively, and
\begin{equation}\label{eq:Ljdef}
L_j = \langle b_j, \nabla_{x,\xi} \rangle 
\end{equation}
is a first order differential operator with constant coefficients such that $b_j \in N(\Delta), j=1,2,\dots,k$ and $k,N \in \no$. 

Observe that in \eqref{eq:kernelchar2} we may substitute $N(\Delta^\perp)$ by any linear subspace transversal to $N(\Delta)$, that is any vector subspace $V \subseteq T^* \rr{2d}$ such that $T^* \rr{2d} = N(\Delta) \oplus V$. 
Note also that 
\begin{equation*}
\frac{1}{2} \la x_1-x_2, \xi_1-\xi_2 \ra = \langle \pi_{\Delta^\perp} x ,\xi \rangle.
\end{equation*}

In the following we generalize \eqref{eq:kernelchar2} by replacing the diagonal $\Delta$ by a general linear subspace, and the dimension $2d$ is replaced by $d$. 
For simplicity of notation we work with $\rho=1$ but this can be generalized to $0 \leqs \rho \leqs 1$.

\begin{defn}\label{def:Gconormal}
Suppose $Y \subseteq \rr d$ is an $n$-dimensional linear subspace, $0 \leqs n \leqs d$, let $N(Y) = Y \times Y^\perp$, 
and let $V \subseteq T^* \rr d$ be a $d$-dimensional linear subspace such that $N(Y) \oplus V = T^* \rr d$. 
Then $u \in \cS'(\rr d)$ is $\Gamma$-conormal to $Y$ of degree $m\in \ro$, denoted $u \in I^m_\Gamma(\rr d,Y)$, if for some $g \in \cS(\rr d) \setminus 0$ and for any $ k,N \in \mathbb{N}$ we have
\begin{equation}
\label{eq:conormchar}
\begin{aligned}
\left| L_1 \cdots L_k \cT^Y_g u (x,\xi) \right |
& \lesssim \left( 1 + \dist((x,\xi),V) \right)^{m-k} \left( 1 + \dist((x,\xi),N(Y)) \right)^{-N}, \\
& \qquad (x,\xi) \in T^* \rr d, 
\end{aligned}
\end{equation}
where
\begin{equation*}
\cT_g^Y u(x,\xi) = e^{-i \la \pi_{Y^\perp} x, \xi\ra} \cT_g u (x, \xi), \quad (x,\xi) \in T^* \rr d, 
\end{equation*}
and
$L_j$, $j=1,\dots,k$, are first order differential operators defined by \eqref{eq:Ljdef} with $b_j \in N(Y)$. 
\end{defn}

For a fixed $g \in \cS \setminus 0$ we equip $I^m_\Gamma(\rr d,Y)$ with a topology using seminorms defined as the best possible constants in \eqref{eq:conormchar} for $N,M \in \no$ fixed, maximized over $k\leqs M$ and all combinations of $b_j \in N(Y)$ belonging to a fixed and arbitary basis.

As observed, the definition is independent of the linear subspace $V$ as long as $N(Y) \oplus V = T^* \rr d$, and often it is convenient to use $V = N(Y)^\perp = N(Y^\perp)$. 
We will also see that the definition and the topology does not depend on $g \in \cS(\rr d) \setminus 0$ (see Corollary \ref{cor:windowindep2}). 

If we pick coordinates such that $Y = \rr n \times \{0\} \subseteq \rr d$ then 
\begin{align*}
N(Y) & = \{(x_1, 0, 0, \xi_2): \ x_1 \in \rr n, \, \xi_2 \in \rr {d-n}\} \subseteq T^* \rr d, \\
N(Y^\perp) & = \{(0, x_2, \xi_1, 0): \ x_2 \in \rr {d-n}, \, \xi_1 \in \rr {n}\} \subseteq T^* \rr d. 
\end{align*} 
We split variables as $x=(x_1,x_2) \in \rr d$, $x_1 \in \rr n$, $x_2 \in \rr {d-n}$. 
The inequalities \eqref{eq:conormchar} reduce to 
\begin{equation}
\label{eq:Gconormdefineq}
|\partial^\alpha_{x_1} \partial^\beta_{\xi_2} \left( e^{-i \la x_2, \xi_2 \ra }\cT_g u (x,\xi) \right)| \lesssim \eabs{(x_1,\xi_2)}^{m-|\alpha+\beta|} \eabs{(x_2,\xi_1)}^{-N}
\end{equation}
for $\alpha \in \nn n$, $\beta \in \nn {d-n}$ and $N \in \no$.

\begin{example}
By Proposition \ref{prop:LGchar} and \eqref{eq:kernelchar2} we have 
\begin{equation*}
I^m_\Gamma(\rr {2d},\Delta) = \{ K_a \in \cS'(\rr {2d}): a \in \Gamma^m(\rr {2d}) \}. 
\end{equation*}
\end{example}

\begin{example}
Write $x=(x_1,x_2)$, $x_1 \in \rr n$, $x_2 \in \rr {d-n}$, and consider $u = 1 \otimes \delta_0 \in \cS'(\rr d)$ with $1\in\cS'(\rr{n})$ and $\delta_0 \in\cS'(\rr {d-n})$. 
The distribution $u$ is a prototypical example of a distribution $\Gamma$-conormal (and also conormal in the standard sense of \cite[Chapter~18.2]{Hormander0}) to the subspace $\rr n \times \{0\}$. 
It is a Gaussian distribution in the sense of H\"ormander \cite{Hormander2} (cf. \cite{PRW1}). A computation yields
\begin{equation*}
\cTp u(x,\xi) =  (2\pi)^{- \frac{d-n}{2}} \pi^{-\frac{d}{4}}e^{i \la x_2, \xi_2 \ra} e^{-\frac{1}{2} (|x_2|^2+|\xi_1|^2)}
\end{equation*}
so the inequalities \eqref{eq:Gconormdefineq} are satisfied for $m=0$. In particular $\delta_0 (\rr d) \in I_{\Gamma}^0(\rr d, \{ 0 \})$.
\end{example}

Next we characterize the conormal distributions of which the latter example is a particular case. 
Again we denote $x=(x_1,x_2)\in \rr d$, $x_1 \in \rr n$, $x_2 \in \rr {d-n}$. 
\begin{lem}
\label{lem:IGchar}
If $u \in \cS'(\rr d)$ and $0 \leqs n \leqs d$ then $u\in I^m_\Gamma(\rr d, \rr n \times \{0\})$ if and only if
\begin{equation*}
u(x) = (2\pi)^{-(d-n)/2} \int_{\rr {d-n}} e^{i \la x_2,\theta \ra} a(x_1,\theta)\ \dd \theta
\end{equation*}
for some $a\in \Gamma^m(\rr d)$, that is $u=\cF_2^{-1}a$. 
\end{lem}
\begin{proof}
Let $g\in\cS(\rr d)\setminus 0$. By Lemma \ref{lem:Fourier} we have
\begin{equation*}
\cT_g  u(x_1,x_2,\xi_1,\xi_2) = e^{i \la x_2,\xi_2 \ra} \cT_{\cF_2 g} \cF_2  u (x_1,\xi_2,\xi_1, -x_2).
\end{equation*}
Set $a=\cF_2 u\in\cS^\prime(\rr d)$. Proposition \ref{prop:symbchar} implies that $a\in \Gamma^m(\rr d)$ if and only if the estimate \eqref{eq:Gconormdefineq} hold for all for $\alpha \in \nn n$, $\beta \in \nn {d-n}$ and $N \in \no$. 
By Definition \ref{def:Gconormal} this happens exactly when $u\in I^m_\Gamma(\rr d, \rr n \times \{0\})$. 
\end{proof}

The extreme cases $n=0$ and $n=d$ yield

\begin{cor}
\label{cor:extremeconormcases}
$I^m_\Gamma(\rr d, \{0\}) = \cF \Gamma^m(\rr d)$ and $I^m_\Gamma(\rr d, \rr d) = \Gamma^m(\rr d)$. 
\end{cor}

The proof of Lemma \ref{lem:IGchar} gives the following byproduct. 

\begin{cor}\label{cor:windowindep1}
The topology on $I^m_\Gamma(\rr d, \rr n \times \{0\})$ does not depend on $g$. 
\end{cor}

The next result treats how $\Gamma$-conormal distributions behave under orthogonal coordinate transformations.

\begin{lem}\label{lem:conormalcoord}
If $Y \subseteq \rr d$ is an $n$-dimensional linear subspace, $0 \leqs n \leqs d$, 
and $B \in \On(d)$ then $B^*: I_\Gamma^m(\rr{d},Y) \rightarrow I_\Gamma^m(\rr d, B^t Y)$ is a homeomorphism.
\end{lem}

\begin{proof}
Let $g \in \cS(\rr d) \setminus 0$. We have 
\begin{equation*}
\cT_g (B^*u) (x,\xi) = \cT_h u (B x, B \xi) 
\end{equation*}
where $h = (B^t)^*g \in \cS(\rr d)$. 
From this and $\pi_{(B^t Y)^\perp} = B^t \pi_{Y^\perp} B$ we obtain 
\begin{equation*}
\cT_g^{B^t Y} (B^*u) (x,\xi) = \cT_h^Y u (B x, B \xi) 
\end{equation*}
so $B^*u \in I_\Gamma^m(\rr d, B^t Y)$ follows from Definition \ref{def:Gconormal}, $N(B^t Y) = B^t Y \times B^t Y^\perp$ and 
\begin{equation*}
\dist( (Bx,B\xi),N(Y)) = \dist((x,\xi),N(B^t Y)), \qquad (x,\xi) \in T^* \rr d. 
\end{equation*}
It also follows that the map $u \rightarrow B^*u$ is continuous from $I_\Gamma^m(\rr d,Y)$ to $I_\Gamma^m(\rr d,B^tY)$ when the topologies for $I_\Gamma^m(\rr d,Y)$ and $I_\Gamma^m(\rr d,B^tY)$
are defined by means of $h \in \cS$ and $g \in \cS$, respectively. 
\end{proof}

If we combine Lemma \ref{lem:conormalcoord} with Corollary \ref{cor:windowindep1} then we obtain the following generalization of the latter result. 

\begin{cor}\label{cor:windowindep2}
If $Y \subseteq \rr d$ is an $n$-dimensional linear subspace, $0 \leqs n \leqs d$, 
then the topology on $I^m_\Gamma(\rr d, Y)$ does not depend on $g$. 
\end{cor}

We can also extract the following generalization of Lemma \ref{lem:IGchar} from Lemma \ref{lem:conormalcoord}.
\begin{prop}
\label{prop:IGchar}
Let $0 \leqs n \leqs d$ and let $Y \subseteq \rr d$ be an $n$-dimensional linear subspace. 
Then $u \in \cS'(\rr d)$ satisfies $u \in I^m_\Gamma(\rr d,Y)$ if and only if 
\begin{equation}\label{uoscint}
u(x) = \int_{\rr {d-n}} e^{i \la M_2^t x, \theta \ra} a(M_1^t x, \theta) \, \dd \theta
\end{equation}
for some $a \in \Gamma^m(\rr d)$, where $M_2 \in \M_{d \times (d- n)}( \ro)$ and $M_1 \in \M_{d \times n}( \ro)$ are matrices such that 
$Y = \Ker M_2^t$ and 
$U = [M_1 \ M_2] \in \GL(d,\ro)$. 
\end{prop}
\begin{proof}
If $u \in I^m_\Gamma(\rr d,Y)$ then we can pick $U = [M_1 \ M_2] \in \On(d)$ 
where $M_1 \in \M_{d \times n}(\ro)$ and $M_2 \in \M_{d \times (d-n)}(\ro)$ such that $Y = \Ker M_2^t$, 
which implies that $U^t Y = \rr n \times \{0\}$. 
By Lemma \ref{lem:conormalcoord} we have $U^* u \in I^m_\Gamma(\rr d,\rr n \times \{0\})$, 
and  \eqref{uoscint} with $a \in \Gamma^m(\rr d)$ is then a consequence of Lemma \ref{lem:IGchar}. 

Suppose on the other hand that \eqref{uoscint} holds for $a \in \Gamma^m(\rr d)$ and $U = [M_1 \ M_2] \in \GL(d,\ro)$. Set $Y = \Ker M_2^t$. 
We may assume that $U = [M_1 \ M_2] \in \On(d)$, after modifying $a \in \Gamma^m(\rr d)$ by means of a linear invertible coordinate transformation, which is permitted since $\Gamma^m$ is invariant under such transformations. 
By Lemma \ref{lem:IGchar} we have $U^* u \in I^m_\Gamma(\rr d,\rr n \times \{0\})$, 
and Lemma \ref{lem:conormalcoord} then gives $u \in I^m_\Gamma(\rr d,Y)$.
\end{proof}

Since 
\begin{equation*}
\bigcap_{m \in \ro} \Gamma^m(\rr d) = \cS(\rr d)
\end{equation*}
we have the following consequence. 

\begin{cor}
If $0 \leqs n \leqs d$ and $Y \subseteq \rr d$ is an $n$-dimensional linear subspace then 
\begin{equation*}
\cS(\rr d) \subseteq I^m_\Gamma(\rr d,Y). 
\end{equation*}
\end{cor}

We also obtain a generalization of Lemma \ref{lem:conormalcoord}. 

\begin{cor}
\label{cor:coordchange}
If $Y \subseteq \rr d$ is an $n$-dimensional linear subspace, $0 \leqs n \leqs d$, 
and $B \in \GL(d,\ro)$ then $B^*: I_\Gamma^m(\rr{d},Y) \rightarrow I_\Gamma^m(\rr d, B^{-1} Y)$ is a homeomorphism.
\end{cor}

\begin{proof}
By Proposition \ref{prop:IGchar} we have $u \in I^m_\Gamma(\rr d,Y)$ if and only if $B^*u \in I^m_\Gamma(\rr d,B^{-1}Y)$.
It remains to show that $B^*$ is continuous. 
By Lemma \ref{lem:conormalcoord} we may replace $Y$ with any $n$-dimensional linear subspace.
Using the singular value decomposition $B = U \Sigma V^t$, where $U,V \in \On(d)$ and $\Sigma$ is diagonal with positive entries, 
the proof of the continuity of $B^*$ reduces, again using Lemma \ref{lem:conormalcoord}, to a proof of the continuity of 
\begin{equation*}
\Sigma^*: I^m_\Gamma(\rr d,\rr n \times \{0\}) \rightarrow I^m_\Gamma(\rr d,\rr n \times \{0\}). 
\end{equation*}
The latter continuity follows straightforwardly using the estimates
\eqref{eq:Gconormdefineq}. 
\end{proof}

By Lemma \ref{lem:Fourier}
\begin{equation*}
\cT_{\wh g} \wh  u (x,\xi)
= e^{i \la x,\xi \ra}  \cT_g  u(-\xi,x) 
\end{equation*}
which gives 
\begin{align*}
\cT_{\wh g}^{Y^\perp} \wh  u (x,\xi)
= e^{i (\la x,\xi \ra - \la \pi_Y x,\xi \ra)}  \cT_g  u(-\xi,x) 
= \cT_g^{Y}  u(-\xi,x). 
\end{align*}
Thus it follows from Definition \ref{def:Gconormal} that 
$\cF: I^m_\Gamma(\rr{d},Y) \rightarrow I^m_\Gamma(\rr{d},Y^\perp)$ continuously. 

\begin{prop}
\label{prop:FourierImg}
If $Y \subseteq \rr d$ is an $n$-dimensional linear subspace, $0 \leqs n \leqs d$, 
then the Fourier transform is a homeomorphism from $I^m_\Gamma(\rr{d},Y)$ to $I^m_\Gamma(\rr{d},Y^\perp)$.
\end{prop}

\begin{example}
If $u \in I_\Gamma^m(\rr d, \rr n \times \{0\})$ then by Lemma \ref{lem:IGchar} there exists $a \in \Gamma^m(\rr d)$ such that 
\begin{equation*}
u(x) = (2\pi)^{-(d-n)/2} \int_{\rr {d-n}} e^{i \la x_2,\theta \ra} a(x_1,\theta)\ \dd \theta. 
\end{equation*}
If $B \in \GL(d,\ro)$ and 
\begin{equation*}
B=\begin{pmatrix}
B_1 & 0 \\
0 & B_2
\end{pmatrix}
\end{equation*} 
then the action of $B$ can understood as an action on the symbol of $u$, 
\begin{equation*}
B^*u(x) = (2\pi)^{-(d-n)/2} \int_{\rr {d-n}} e^{i \la x_2,\theta \ra} a(B_1 x_1,B_2^{-t}\theta) |B_2 |^{-1} \, \dd \theta.
\end{equation*}
\end{example}

\begin{rem}\label{rem:conormalgeom}
The estimates \eqref{eq:conormchar} in Definition \ref{def:Gconormal} can be translated to a geometric form, as in Remark \ref{rem:kernelchargeom} for Schwartz kernels of Shubin operators. 
The result is 
\begin{align*}
& \left| (\Pi_{N(Y)} (x,\xi))^\alpha  (\Pi_{N(Y)} \partial_{x,\xi})^\beta \cT^Y_g u (x,\xi) \right | \\
& \qquad \lesssim \left( 1 + \dist((x,\xi),V) \right)^{m} \left( 1 + \dist((x,\xi),N(Y)) \right)^{-N}, 
\end{align*}
for $\alpha, \beta \in \nn {2d}$ such that $|\alpha|=|\beta|$, and $N \in \no$ arbitrary. 
\end{rem}

\begin{rem}
Let $X$ be a smooth manifold of dimension $d$ and let $Y \subseteq X$ be a closed submanifold. 
H\"ormander's conormal distributions $I^m(X,Y)$ with respect to $Y$ of order $m\in \ro$ is by \cite[Definition~18.2.6]{Hormander0} all $u\in \mathcal{D}'(X)$ such that 
\begin{equation*}
L_1\dots L_k u \in B^{-m-d/4}_{2,\infty, \, \rm loc}(X), \quad k \in \no, 
\end{equation*}
where $L_j$ are first order differential operators with coefficients tangential to $Y$, 
and where $B^{-m-d/4}_{2,\infty, \, \rm loc}(X)$ is a Besov space.

Comparing this definition with the estimates defining $I_\Gamma^m(\rr d,Y)$ in Remark \ref{rem:conormalgeom} we see that the fact that we are working with isotropic symbol classes made it necessary to replace the local, Fourier-based Besov spaces with a global, isotropic version based on the transform $\cTp$, resembling a modulation space. 

We note that he submanifold $Y$ is allowed to be nonlinear in $I^m(X,Y)$, as opposed to the linear 
submanifold $Y \subseteq \rr d$ we use in $\Gamma$-conormal distributions $I_\Gamma^m(\rr d,Y)$. 
\end{rem}

\subsection{Microlocal properties of $\Gamma$-conormal distributions}

The wave front set of a conormal distribution in $I^m(X,Y)$ is contained in the conormal bundle of the submanifold $Y$ \cite[Lemma~25.1.2]{Hormander0}. 

The wave front set adapted to the Shubin calculus is the Gabor wave front set studied e.g. in \cite{Hormander1, Nakamura1, Rodino1,SW,SW2}, see also \cite{CS}.
It can be introduced using either pseudodifferential operators or the short-time Fourier transform. In the latter definition one may replace $\mathcal{V}_g u$ by $\cT_g u$ since they are identical up to a factor of modulus one.
\begin{defn}
\label{def:WFG}
If $u \in \cS'(\rr d)$ and $g\in\cS(\rr d)\setminus0$ then $(x_0,\xi_0)\in T^*\rr{d} \setminus{0}$ satisfies  $(x_0,\xi_0) \notin \WF_G(u)$ if 
there exists an open cone $V \subseteq T^* \rr d \setminus 0$ containing $(x_0,\xi_0)$, such that for any $N\in\no$ there exists $C_{V,g,N}>0$ such that $|\cT_g u(x,\xi)|\leqs C_{V,g,N} \eabs{(x,\xi)}^{-N}$ when $(x,\xi) \in V$.
\end{defn}
The definition does not depend on $g\in\cS(\rr d)\setminus0$. 
The Gabor wave front set transforms well under the metaplectic operators discussed in Section \ref{sec:prelim}, cf. \cite{Hormander1}, that is
\begin{equation*}
\WF_G(\mu(\chi) u) = \chi \left(\WF_G(u)\right), \quad u \in \cS'(\rr d), \quad \chi \in \Sp(d,\ro). 
\end{equation*}
\begin{prop}\label{prop:WFconormal}
Let $Y \subseteq \rr d$ be an $n$-dimensional linear subspace, $0 \leqs n \leqs d$. 
If $u \in I^m_\Gamma(\rr d,Y)$ then 
\begin{equation*}
WF_G( u) \subseteq N(Y).  
\end{equation*}
\end{prop}
\begin{proof}
Suppose $(x,\xi) \notin N(Y)$. This means $(\pi_{Y^\perp} x,\pi_Y\xi) \neq 0$, so $(x,\xi) \in V$ where
the open conic set $V \subseteq T^* \rr d$ is defined by 
\begin{equation*}
V = \{ (x,\xi) \in T^* \rr d : \ |(\pi_Y x, \pi_{Y^\perp} \xi)| < C |( \pi_{Y^\perp} x,\pi_Y \xi )| \}
\end{equation*}
for some $C>0$. 
Using 
\begin{equation*}
|(x,\xi)|^2 = |(\pi_Yx, \pi_{Y^\perp} \xi)|^2 + |( \pi_{Y^\perp} x,\pi_Y\xi)|^2, 
\end{equation*}
$\dist(x,Y) = | \pi_{Y^\perp} x|$, $\dist(x,Y^\perp) = |\pi_Yx|$ and 
\begin{equation*}
\dist^2( (x,\xi), N(Y) ) = \dist^2( x ,Y) + \dist^2(\xi,Y^\perp), 
\end{equation*}
the result follows from Definition \ref{def:Gconormal} (with trivial operators $L_j$). 
\end{proof}

\begin{cor}\label{cor:WFpsdokernel}
If $a \in \Gamma^m(\rr {2d})$ and $a^w(x,D)$ has Schwartz kernel $K_a$ then
\begin{equation*}
\WF_G(K_a) \subseteq N(\Delta) \subseteq T^* \rr {2d}.  
\end{equation*}
\end{cor}

It is well known that Shubin pseudodifferential operators are microlocal with respect to $\WF_G$, that is if $a \in \Gamma^m( \rr {2d})$ and $u \in \cS' (\rr d)$ then
\begin{equation*}
\WF_G(a^w(x,D) u)\subseteq \WF_G(u),
\end{equation*}
see e.g. \cite{Hormander1,SW2}. We show that they also preserve $\Gamma$-conormality.

\begin{prop}
\label{prop:pseudocomp}
Let $Y \subseteq \rr d$ be an $n$-dimensional linear subspace, $0 \leqs n \leqs d$. 
If $a \in \Gamma^{m'}(\rr {2d})$ then $a^w(x,D)$ is continuous from $I^m_\Gamma(\rr d,Y)$ to $I_\Gamma^{m+m'}(\rr d,Y)$.
\end{prop}

\begin{proof}
If $a \in \Gamma^{m'} (\rr {2d})$ and $U \in \On(d)$ then we have by symplectic invariance of the Weyl calculus 
\eqref{metaplecticoperator}
\begin{equation*}
(U^t)^* a^w(x,D) U^* = b^w(x,D)
\end{equation*}
where $b(x,\xi) = a(U^t x, U^t \xi) \in \Gamma^{m'} (\rr {2d})$. By Lemma \ref{lem:conormalcoord} we may therefore assume that $Y = \rr n \times \{0\}$. The symplectic invariance also guarantees that
\begin{equation*}
\mathscr{F}_2^{-1} b^w(x,D) \mathscr{F}_2 = c^w(x,D)
\end{equation*}
with $c(x,\xi) = b(x_1,\xi_2,\xi_1,-x_2)\in \Gamma^{m'} (\rr {2d})$ 
where $x=(x_1,x_2)\in \rr d$, $x_1 \in \rr n$, $x_2 \in \rr {d-n}$. 
To prove $a^w(x,D) u \in I_\Gamma^{m+m'}(\rr d,\rr n \times \{0\})$ for $a \in \Gamma^{m'} (\rr {2d})$ and $u\in I^m_\Gamma(\rr d,\rr n \times \{0\})$
is therefore by Lemma \ref{lem:IGchar} equivalent to proving that $a^w(x,D) u \in \Gamma^{m+m'}(\rr d)$ for $a \in \Gamma^{m'} (\rr {2d})$ and $u \in \Gamma^m(\rr d)$. 

Let $a \in \Gamma^{m'} (\rr {2d})$, $u \in \Gamma^m(\rr d)$ and set $A= a^w(x,D)$. By Proposition \ref{prop:symbchar} it suffices to verify
\begin{equation*}
|\partial^\alpha_{x} \cTp Au (x,\xi)| \lesssim \eabs{x}^{m+m'-|\alpha|}\eabs{\xi}^{-N}, \quad (x,\xi) \in T^* \rr d, 
\end{equation*}
for any $N \geqs 0$ and $\alpha \in \nn d$. 

Let $N \geqs 0$ and $\alpha \in \nn d$. 
Writing $\cTp Au=(\cTp A\cTp^*)\cTp u$ and using \eqref{eq:cTpkernelident} we are thus tasked with estimating $\partial^\alpha_{x}$ acting on 
\begin{equation}\label{eq:kernelreformulation}
\begin{aligned}
\cTp Au (x,\xi) & = \int_{\rr {2d}} \cTp K_a(x,y,\xi,-\eta) \cTp u(y,\eta)\ \dd y \, \dd \eta \\
& = \int_{\rr {2d}} e^{\frac{i}{2} \la x-y,\xi+\eta \ra} \, \cTp^\Delta K_a(x,y,\xi,-\eta) \, \cTp u(y,\eta)\ \dd y \, \dd \eta.
\end{aligned}
\end{equation}

The integral \eqref{eq:kernelreformulation} converges due to the estimates
\begin{equation*}
\label{eq:Gineqker}
|\partial_y^\alpha \cTp u(y,\eta)|\lesssim \eabs{y}^{m-|\alpha|}\eabs{\eta}^{-N}, \quad y, \eta \in \rr d, \quad \alpha \in \nn d, \quad N \geqs 0, 
\end{equation*}
which follows from Proposition \ref{prop:symbchar}, and the estimates
\begin{equation*}
\begin{aligned}
| (\partial_x + \partial_y)^\alpha \cTp ^\Delta K_a (x,y, \xi, -\eta)| 
& \lesssim \eabs{(x+y,\xi + \eta)}^{m'-|\alpha|} \eabs{(x-y,\xi-\eta)}^{-N}, \\
& \qquad x,y,\xi,\eta \in \rr d, \quad \alpha \in \nn d, \quad N \geqs 0, 
\end{aligned}
\end{equation*}
that are guaranteed by Proposition \ref{prop:LGchar}. 

Writing $\partial_{x_j} = \partial_{x_j} + \partial_{y_j} - \partial_{y_j}$ for $1 \leqs j \leqs d$ and differentiating under the integral in \eqref{eq:kernelreformulation} we obtain by integration by parts for any $N_1,N_2 \geqs 0$
\begin{equation*}
\begin{aligned}
& \left|\partial^\alpha_{x} \cTp Au (x,\xi)\right| \\
& = \sum_{\beta \leqs \alpha}  C_{\beta} \left| \int_{\rr {2d}} (\partial_x + \partial_y)^\beta \left(e^{\frac{i}{2} \la x-y,\xi+\eta \ra}\,\cTp^\Delta K_a(x,y,\xi,-\eta)\right) \, \partial^{\alpha-\beta}_y \cTp u(y,\eta)\ \dd y \, \dd \eta\right|\\
& = \sum_{\beta \leqs \alpha}  C_{\beta} \left| \int_{\rr {2d}} e^{\frac{i}{2} \la x-y,\xi+\eta \ra} \, (\partial_x + \partial_y)^\beta \, \cTp^\Delta K_a(x,y,\xi,-\eta) \, \partial^{\alpha-\beta}_y \cTp u(y,\eta)\ \dd y \, \dd \eta\right|\\
& \lesssim \sum_{\beta \leqs \alpha} \int_{\rr {2d}} \left|(\partial_x + \partial_y)^\beta \cTp^\Delta K_a(x,y,\xi,-\eta) \, \partial^{\alpha-\beta}_y \cTp u(y,\eta)\right|\ \dd y \, \dd \eta,\\
& \lesssim \sum_{\beta \leqs \alpha}  \int_{\rr {2d}} \eabs{(x+y,\xi + \eta)}^{m'-|\beta|} \eabs{(x-y,\xi-\eta)}^{-N_1} \, \eabs{y}^{m-|{\alpha-\beta}|}\eabs{\eta}^{-N_2} \, \dd y \, \dd \eta.
\end{aligned}
\end{equation*}

Finally we estimate 
\begin{align*}
\int_{\rr {2d}} &\eabs{(x+y,\xi+\eta)}^{m'-|\beta|} \eabs{(x-y,\xi-\eta)}^{-N_1}\eabs{y}^{m-|\alpha-\beta|} \eabs{\eta}^{-N_2} \dd y \, \dd \eta\\
& = \int_{\rr {2d}}  \eabs{(2x+y,2\xi+\eta)}^{m'-|\beta|}\eabs{(y,\eta)}^{-N_1}\eabs{y+x}^{m-|\alpha-\beta|} \eabs{\eta+\xi}^{-N_2} \dd y \, \dd \eta\\
&\lesssim \int_{\rr {2d}}  \eabs{x}^{m'-|\beta|} \eabs{y}^{|m'|+|\beta|}\eabs{\xi}^{|m'|+|\beta|}\eabs{\eta}^{|m'|+|\beta|}\eabs{(y,\eta)}^{-N_1}\eabs{x}^{m-|\alpha-\beta|} \\
& \qquad \qquad \qquad \qquad \qquad \times \eabs{y}^{|m|+|\alpha|} \eabs{\xi}^{-N_2} \eabs{\eta}^{N_2}\dd y \, \dd \eta\\
&\lesssim \eabs{x}^{m'+m-|\alpha|} \eabs{\xi}^{|m'|+|\alpha|-N_2} \int_{\rr {2d}} \eabs{y}^{|m'|+|m| + 2 |\alpha|} \eabs{\eta}^{|m'|+|\alpha|+N_2}\eabs{(y,\eta)}^{-N_1} \dd y \, \dd \eta\\
&\lesssim \eabs{x}^{m'+m-|\alpha|} \eabs{\xi}^{-N},
\end{align*}
provided
$N_1 > N_2+2|m'|+|m|+3|\alpha| +2d$ and $N_2 \geqs N+|m'| + |\alpha|$
This proves 
\begin{equation*}
\left|\partial^\alpha_{x} \cTp Au (x,\xi)\right|\lesssim \eabs{x}^{m^\prime+m-|\alpha|} \eabs{\xi}^{-N}, \quad (x,\xi) \in T^* \rr d
\end{equation*}
and as a by-product of these estimates we obtain the claimed continuity.
\end{proof}

\begin{rem}
The proof shows that the result can be generalized. If $a \in \Gamma_\rho^{m'}(\rr {2d})$ and $u \in I_{\Gamma,\rho}^m(\rr d, Y)$ 
then $a^w(x,D) u \in I_{\Gamma,\rho}^{m+m'}(\rr d, Y)$, for $0 \leqs \rho \leqs 1$. 
Here $I_{\Gamma,\rho}^m(\rr d, Y)$ is defined as in Definition \ref{def:Gconormal} with the modified estimate
\begin{equation*}
\left( 1 + \dist((x,\xi),V) \right)^{m- \rho k} \left( 1 + \dist((x,\xi),N(Y)) \right)^{-N}
\end{equation*}
in \eqref{eq:conormchar}. 
\end{rem}
Since Proposition \ref{prop:pseudocomp} shows how $\Gamma$-conormality is preserved under the action of a pseudodifferential operator, we obtain the following result on conormal elliptic regularity:

\begin{cor}[Conormal elliptic regularity]
Suppose $u \in \cS'(\rr d)$ solves the pseudodifferential equation $a^w(x,D) u = f$ with $f \in I_\Gamma^m(\rr d, Y)$ where $a \in \Gamma^{m'}(\rr {2d})$ is globally elliptic, that is satisfying
\begin{equation}\label{glell}
|a(x,\xi)| \geqs C \langle (x,\xi) \rangle^{m'}, \qquad |(x,\xi)| \geq R 
\end{equation} 
for $C,R>0$. Then $u \in I_\Gamma^{m-m'}(\rr d, Y)$.
\end{cor}
\begin{proof}
Under condition \eqref{glell}, $a^w(x,D)$ admits a parametrix $p^w(x,D)$ with $p \in \Gamma^{-m'}$ and $p^w(x,D)a^w(x,D) = I+R$, where $R$ is continuous $\cS' \rightarrow \cS$ \cite{Shubin1}. Then $u=p^w(x,D)f-Ru$ and hence $u \in I_\Gamma^{m-m'}(\rr {d}, Y)$.
\end{proof}

%%%%%%%%%%%%%%%%%%%%%%%%%%
\section*{acknowledgements}
%%%%%%%%%%%%%%%%%%%%%%%%%%

The authors would like to express their gratitude to Luigi Rodino, Joachim Toft and Moritz Doll for helpful discussions on the subject. 

%%%%%%%%%%%%%%%%%%%%%%%%%%%%%

\end{document}